\documentclass[12pt]{article}
\usepackage{amsmath,amsthm,amssymb}
\usepackage{array, longtable}
\newtheorem{Theorem}{Theorem}[section]
\newtheorem{lem}[Theorem]{Lemma}
\newtheorem{Remark}[Theorem]{Remark}
\newtheorem{Definition}[Theorem]{Definition}
\newtheorem{Corollary}[Theorem]{Corollary}

\newtheorem{Example}[Theorem]{Example}

\numberwithin{equation}{section}
\numberwithin{table}{section}

\begin{document}

\title{Generalized $b$-symbol weights of Linear Codes and $b$-symbol MDS Codes
\footnote{ E-Mail addresses:  hwliu@mail.ccnu.edu.cn (H. Liu), panxu@mails.ccnu.edu.cn (X. Pan)}}

\author{Hongwei Liu,~Xu Pan}
\date{\small
School of Mathematics and Statistics, Central China Normal University\\Wuhan, Hubei, 430079, China\\
}
\maketitle

\begin{abstract}
Generalized pair weights of linear codes are generalizations of minimum symbol-pair weights, which were introduced by Liu and Pan \cite{LP} recently. Generalized pair weights can be used to characterize the ability of protecting information in the symbol-pair read wire-tap channels of type II. In this paper, we introduce the notion of generalized $b$-symbol weights of linear codes over finite fields, which is a generalization of generalized Hamming weights and generalized pair weights. We obtain some basic properties and bounds of generalized $b$-symbol weights which are called Singleton-like bounds for generalized $b$-symbol weights. As examples, we calculate generalized weight matrices for simplex codes and Hamming codes.
We provide a necessary and sufficient condition for a linear code to be a $b$-symbol MDS code by using the generator matrix and the parity check matrix of this linear code.
Finally, a necessary and sufficient condition of a linear isomorphism preserving  $b$-symbol weights between two linear codes is obtained. As a corollary, we get the classical MacWilliams extension theorem when $b=1$.


\end{abstract}

\section{Introduction}

In 2011, motivated by the limitations of the reading process in high density data storage systems,  Cassuto and Blaum \cite{CB} introduced a new metric framework, named symbol-pair distance, to protect against pair errors in symbol-pair read channels, where the outputs are overlapping pairs of symbols.
They \cite{CB} also established relationships between the minimum Hamming distance and the minimum pair distance of the code, and obtained lower and upper bounds on the code sizes by using symbol-pair distance.
In \cite{C}, the authors established a Singleton-like bound for symbol-pair codes and constructed MDS symbol-pair codes (meeting this Singleton-like bound).
Several works have been done on the constructions of MDS symbol-pair codes (see, for example, \cite{CLL}, \cite{CJK}, \cite{D}, \cite{KZL}, \cite{KZZLC}, \cite{LG} \cite{ML1} and \cite{ML2}).
In \cite{DNSS} and \cite{DWLS}, the authors calculated the symbol-pair distances of repeated-root constacyclic codes of lengths $p^s$ and $2p^s$, respectively. In 2016, Yaakobi, Bruck and Siegel \cite{YBS} generalized the notion of symbol-pair weight to $b$-symbol weight. Yang, Li and Feng \cite{YL} showed the Plotkin-like bound for the $b$-symbol weight and presented a construction on irreducible cyclic codes and constacyclic codes meeting the Plotkin-like bound.

On the other hand, the notion of generalized Hamming weights appeared in the 1970's and has become an important research object in coding theory after Wei's work \cite{W} in 1991. Wei \cite{W} showed that the generalized Hamming weight hierarchy of linear codes has a close connection with cryptography.
Since then, lots of works have been done in computing and describing the generalized Hamming weight hierarchies of certain linear codes (see, for example, \cite{B},  \cite{JFW}, \cite{TV} and \cite{YF}).
In \cite{LP}, we introduced the notion of generalized pair weights of linear codes, which is a generalization of minimum symbol-pair weights of linear codes, we obtained some bounds for generalized pair weights and gave an application of generalized pair weights of linear codes to symbol-pair read wire-tap channels of type II.


It is well-known that the  MacWilliams extension theorem plays a central role in coding theory.
MacWilliams \cite{M} and later Bogart, Goldberg, and Gordon
\cite{BGG} proved that, every linear isomorphism preserving Hamming weights between two linear codes over finite fields can be induced by a monomial matrix. It is interesting to ask how about the behavior of every linear isomorphism preserving the $b$-symbol weights between two linear codes.  Unfortunately, we found that a linear isomorphism induced by a permutation matrix may not preserve the $b$-symbol weight between two linear codes. In \cite{LP}, the authors  provided a necessary and sufficient condition for a linear isomorphism preserving pair weights between two linear codes.

 In 2018, Ding, Zhang and Ge \cite{DZ} established the Singleton-like bound $$d_{b}(C)\leq n+b-k$$ for an $[n,k]$-linear code $C$, where $d_{b}(C)$ is the minimum $b$-symbol weight of $C$ defined in Section 2. Since no linear code exists that reaches this bound in \cite{DZ} when $b>k$, we give an improvement for the Singleton-like bound $$d_{b}(C)\leq  \min\{n+b-k,n\}$$ in Theorem \ref{pair monotonicity}, which is a small part of Singleton-like bound for generalized weight matrices of linear codes in Theorem \ref{GBW}.
 Linear codes meeting this bound are called $b$-symbol MDS codes in this paper (see Def.~\ref{bmds}).
 Then we show that the length of $b$-symbol MDS codes is as large as possible when $b\ge k$ in Example \ref{4.4}, which is different from MDS conjecture that the length of $1$-symbol MDS codes is less than or equal to $q+1$ or $q+2$ (some special cases). Let $n_{r,k}$ denote the number of all $r$-dimensional subspaces of a $k$-dimensional vector space. It is interesting that when we study the $b$-symbol weights of linear codes and $b$-symbol MDS codes, we found that the length $n\leq n_{1,b+1}$ for any $[n,k]$-linear $b$-symbol MDS code over $\mathbb{F}_{q}$ if $b=k-1$ in Corollary~\ref{5.3c}. And MDS conjecture is that if $b=1$ then $n\leq n_{1,b+1}$ for any nontrivial $[n,k]$-linear $b$-symbol MDS code over $\mathbb{F}_{q}$ except $q$ is even and $k=3$ or $k=q-1$.
Hence it is bold to conjecture that $$n\leq n_{1,b+1}$$ for any $[n,k]$-linear $b$-symbol MDS code over $\mathbb{F}_{q}$ except some special cases, for example $q$ is even and $k=3$ or $k=q-1$.

In this paper, we unify the works of \cite{LP} and \cite{W} to introduce generalized $b$-symbol weights of linear codes for $1\leq b \leq n$. We define the generalized weight matrix $D(C)$ of a linear code $C$ in Section 2. The parameters about generalized $b$-symbol weights of the linear code $C$ for $1\leq b\leq n$ can be obtained from the generalized weight matrix $D(C)$.
The results about the generalized weight matrix $D(C)$ is proved in Theorem \ref{GBW}. And we calculate the generalized weight matrices $D(C)$ of simplex codes and two especial Hamming codes in Section 4.
In Section 6, we provide a necessary and sufficient condition of a linear isomorphism preserving $b$-symbol weights between two linear codes. As a corollary, when $b=1$, we obtain the classical MacWilliams extension theorem. Then we provide an algorithm to determine whether an isomorphism between two linear codes preserves $b$-symbol weights by using this theorem. And we explain why this algorithm is more efficiently than simply checking $b$-symbol weights of all the codewords of two codes in Remark \ref{5.6}.



This paper is organized as follows.
Section 2 provides some preliminaries. We introduce the notion of generalized $b$-symbol weights of linear codes and give a characterization of the $b$-symbol weight of an arbitrary subspace of linear codes.
In Section~3, we give a relationship between generalized Hamming weights and generalized $b$-symbol weights of linear codes and obtain a Singleton-like bound for generalized $b$-symbol weights.
As examples, we calculate the generalized weight matrix $D(C)$ in Section 4, when $C$ is the simplex code or two especial Hamming codes.
 In Section~5, we provide a necessary and sufficient condition for a linear code to be a $b$-symbol MDS code by using generator matrix and parity check matrix of this linear code.
 In Section~6, we study linear isomorphisms preserving $b$-symbol weights of linear codes, and obtain a necessary and sufficient condition of a linear isomorphism preserving $b$-symbol weights.

\section{Preliminaries}

Throughout this paper, let $\mathbb{F}_q$ be the finite field of order $q$, where $q=p^e$ and $p$ is a prime. And let $\mathbb{N}=\{0, 1,2,\cdots\}$ be the set of all natural numbers and $\mathbb{N}^+=\mathbb{N}\setminus\{0\}$.  A $\mathbb{F}_q$-subspace $C$ of dimension $k$ of $\mathbb{F}_q^{n}$ is called an $[n,k]$-linear code for $k\leq n\in \mathbb{N}^+$. The dual code $C^{\perp}$ of $C$ is defined as
$$
C^{\perp}=\{{\bf x}\in \mathbb{F}_q^n \,| \,  {\bf c}\cdot{\bf x} =0, \forall \, {\bf c}\in C\},
$$
where $ ``-\cdot-" $ is the standard Euclidean inner product.

For $n,b\in \mathbb{N}^+$, we always assume $ 1\leq b \leq n$ in this paper.

\begin{Definition}\label{AA}(\cite{YBS})
For any $\mathbf{x},\mathbf{y} \in \mathbb{F}_{q}^{n}$, the  $b$-symbol distance between $\mathbf{x}$ and $\mathbf{y}$ is defined as
$$d_{b}(\mathbf{x},\mathbf{y})=|\{0\leq i \leq n-1|(x_{i},x_{i+1},\cdots,x_{i+b-1})\neq(y_{i},y_{i+1},\cdots,y_{i+b-1})\}|,$$ where the indices are taken modulo $n$. The $b$-symbol weight of $\mathbf{x}$ is defined as $w_{b}(\mathbf{x})=d_{b}(\mathbf{x},\mathbf{0}).$
\end{Definition}

\begin{Definition}\label{b weight}
Let $D$ be an $ \mathbb{F}_q$-subspace of $ \mathbb{F}_q^{n}$, the $b$-symbol support of $D$, denoted by $\chi_{b}(D)$, is $$\chi_{b}(D)=\{0\leq i \leq n-1\,|\,\exists\, \mathbf{x}=(x_{0},\cdots,x_{n-1})\in D, (x_{i},x_{i+1},\cdots,x_{i+b-1})\neq(0,0,\cdots,0)\},$$  where the indices are taken modulo $n$. The $b$-symbol weight of $D$ is defined as $w_{b}(D)=|\chi_{b}(D)|$.

In particular, if $C$ is an $[n,k]$-linear code over $\mathbb{F}_q$, the {\it minimum $b$-symbol weight} of $C$ is defined as
 $$d_{b}(C)=w_b(C)=\min_{\mathbf{c}\ne\mathbf{c'} \in C}\,d_{b}(\mathbf{c},\mathbf{c}')=\min_{{\bf 0}\ne \mathbf{c}\in C}\,w_{b}(\mathbf{c}).$$

\end{Definition}

We denote by $\langle S\rangle$ the $\mathbb{F}_q$-subspace generated by the subset $S\subseteq \mathbb{F}_q^n$.
By Definition~\ref{b weight}, we know that $w_{b}(\mathbf{x})=w_{b}(\langle\mathbf{x}\rangle)$ for any $\mathbf{x}\in \mathbb{F}_q^{n}$. For convenience, we denote $\chi_{b}(\mathbf{x})=\chi_{b}(\langle\mathbf{x}\rangle)$ for any $\mathbf{x}\in \mathbb{F}_q^{n}$.

\begin{Definition}
Let $C$ be an $[n,k]$-linear code over $\mathbb{F}_q$. For $1\leq r \leq k$, the $r$th generalized $b$-symbol weight of $C$ is defined as $d_{b}^{\,r}(C)=min\{w_{b}(D)\,|\,D\leq C, \dim(D)=r\}$.
And the sequence  $$d_{b}^{\,1}(C),\,d_{b}^{\,2}(C),\cdots,d_{b}^{\,k}(C)$$ is called the hierarchy of  generalized $b$-symbol weights of $C$.
\end{Definition}

\begin{Remark}
When $b=1$, the $r$th generalized $1$-symbol weight $d_{1}^{\,r}(C)$ of $C$ is the $r$th generalized Hamming weight of an $[n,k]$-linear code $C$ over $\mathbb{F}_q$  for $1\leq r \leq k$ defined by Wei \cite{W}. Also we know that $d_{1}^{\,1}(C)$ and $d_{2}^{\,1}(C)$ are the minimum Hamming weight and the minimum pair weight of a linear code $C$, respectively.

When $b=n$, $w_n(\mathbf{c})=n$ for any nonzero $\mathbf{c}\in C$ and $d_{n}^{\,r}(C)=n$ for any $1\leq r\leq k$.
If $D$ is an $\mathbb{F}_q$-subspace of $C$ with $ \dim(D)\ge 1$, we have $d_{b}^{\,r}(C)\leq d_{b}^{\,r}(D)$ for any $1\leq r\leq \dim(D)$.
 \end{Remark}

For convenience, we let $d_{b}(C)=d_{b}^{\,1}(C)$ for any linear code $C$ when $r=1$. Since we want to study all parameters about generalized $b$-symbol weights of the linear code $C$ for $1\leq b\leq n$, we introduce the following definition.

\begin{Definition}
For an $[n,k]$-linear code $C$ over $\mathbb{F}_q$, we define an $n\times k$ matrix $D(C)$ over the field  of real numbers as follows: $$D(C)=(d_b^{\,r}(C))_{n\times k}=\left(\begin{array}{cccc}
                    d_1^{\,1}(C) & d_1^{\,2}(C)&\cdots&d_1^{\,k}(C)\\
        d_2^{\,1}(C) & d_2^{\,2}(C)&\cdots&d_2^{\,k}(C)\\
        \vdots & \vdots&\vdots&\vdots\\
        d_n^{\,1}(C) & d_n^{\,2}(C)&\cdots&d_n^{\,k}(C)
                      \end{array}\right)_{n\times k},$$
where $d_b^{\,r}(C)$ is $r$th generalized $b$-symbol weight of $C$,
for $1\leq b\leq n$ and $1\leq r\leq k$. The matrix $D(C)$ is called the {\it generalized weight matrix} of a linear code $C$.
\end{Definition}

 We note that the elements $d_b^{\,r}(C)$ of the matrix $D(C)$ satisfy some certain rules. For example, every row of the matrix $D(C)$ is increasing from left to right, and every column of the matrix $D(C)$ is increasing from up to down. The properties of the generalized weight matrix $D(C)$ will be provided in Theorem \ref{GBW}.



 Let $U$ be an $\mathbb{F}_q$-vector space of dimension $k$, we denote by $U/W$ the quotient space modulo $W$, where $W$ is an $\mathbb{F}_q$-subspace of $U$.
For any $ r, k\in \mathbb{N}$, let
\begin{align*}
   {\rm PG}^{r}(U)=\{V\leq U\,|\,\dim(V)=r\}&,  & {\rm  PG}^{\leq r}(U)=\{V\leq U\,|\,\dim(V)\leq r\}.
\end{align*}
It is trivial that $\dim(\{\mathbf{0}\})=0$ and ${\rm PG}^{0}(U)=\{\mathbf{0}\}$.
Let $n_{r,k}$ denote the number of all $r$-dimensional subspaces of a $k$-dimensional vector space. It is easy to see that
$$n_{r, k}= \left\{ \begin{array}{ll}
1,  & \textrm{if $r=0\ ;$}\\

\prod\limits_{i=0}^{r-1}\frac{q^{k}-q^{i}}{q^{r}-q^{i}},  & \textrm{if $1\leq r \leq k ;$}\\

0,  & \textrm{if $r> k .$}
\end{array} \right.
$$
Let $C$ be an $[n,k]$-linear code with a generator matrix $G=(G_{0},\cdots,G_{n-1})$,  where $G_i$ is the column vector of $G$ for $0\leq i\leq n-1$. For any $V\in {\rm PG}^{\leq b}(\mathbb{F}_q^{k})$, the function $m^b_{G}: {\rm PG}^{\leq b}(\mathbb{F}_q^{k}) \to \mathbb{N}$ is defined as follows:
$$
m^b_{G}(V)=|\{0\leq i\leq n-1\,|\,\langle G_{i},G_{i+1},\cdots,G_{i+b-1}\rangle= V\}|,
$$  where the indices are taken modulo $n$.
By using the function $m^b_{G}$, we define the function $\theta^b_{G}: {\rm PG}^{\leq k}(\mathbb{F}_q^{k}) \to \mathbb{N}$ to be $$\theta^b_{G}(U)=\sum_{V\in {\rm PG}^{\leq b}(U)}m^b_{G}(V)$$ for any $U\in {\rm PG}^{\leq k}(\mathbb{F}_q^{k})$.

For an $[n,k]$-linear code $C$ over $\mathbb{F}_q$ with a generator matrix $G$, we know that for any $1\leq r \leq k $ and an $\mathbb{F}_q$-subspace $D$ of dimension $r$ of $C$, there exists a unique $\mathbb{F}_q$-subspace $\tilde{D}$ of dimension $r$ of $\mathbb{F}_q^k$ such that $D=\tilde{D}G=\{{\bf y}G\,|\,{\bf y}\in \tilde{D} \}$.
Also we know that for any nonzero codeword ${\bf c}\in C$, there exists a unique nonzero vector ${\bf y}\in \mathbb{F}_q^k$ such that ${\bf c}={\bf y}G=({\bf y}G_0,{\bf y}G_1,\cdots, {\bf y}G_{n-1})$, where $G=(G_0,\cdots, G_{n-1})$.

\begin{lem} \label{r pair weight}
Assume the notations given above. Then $w_{b}(D)= n-\theta^b_{G}(\tilde{D}^{\bot})$ for any subspace $D$ of $C$. In particular $w_{b}(\mathbf{c})= n-\theta^b_{G}(\langle \mathbf{y}\rangle^{\bot})$ for any $0\neq\mathbf{c}\in C$.
\end{lem}

\begin{proof}
By the definition of $w_{b}$ and the function $m_{G}^{b}$, we have
\begin{align*}
  w_{b}(D) & =|\{0\leq i \leq n-1\,|\,\exists\,\mathbf{c}=(c_{0},c_{1},\cdots,c_{n-1})\in D,\,(c_{i},c_{i+1},\cdots,c_{i+b-1})\neq (0,0,\cdots,0)\}| \\
   & =n-|\{0\leq i \leq n-1\,|\,\forall\,\mathbf{c}=(c_{0},c_{1},\cdots,c_{n-1})\in D,\,(c_{i},c_{i+1},\cdots,c_{i+b-1})= (0,0,\cdots,0)\}| \\
   &  =n-|\{0\leq i \leq n-1\,|\,\forall\,\mathbf{y}\in \tilde{D},\,\forall\, i\leq j \leq i+b-1, \,\mathbf{y}G_{j}=0\,\,\}|\\
    &  =n-|\{0\leq i \leq n-1\,|\,\langle G_{i},G_{i+1},\cdots,G_{i+b-1}\rangle\subseteq\tilde{D}^{\bot}\}|  \\
    &=n-\sum_{V\in {\rm PG}^{\leq b}(\tilde{D}^{\bot})}|\{0\leq i \leq n-1\,|\,\langle G_{i},G_{i+1},\cdots,G_{i+b-1}\rangle= V\}|\\
&=n-\sum_{V\in {\rm PG}^{\leq b}(\tilde{D}^{\bot})}m^b_{G}(V)\\
    &=n-\theta^b_{G}(\tilde{D}^{\bot}).
\end{align*}
\end{proof}


Assume $\mathbb{Z}_n=\{0,1,\cdots,n-1\}$.
When we study the $b$-supports $\chi_{b}(C)$ and generalized $b$-symbol weights of linear codes $C$, we can view $\chi_{b}(C)$ as a subset of $\mathbb{Z}_n$ and we need the following definition.
\begin{Definition}\label{hole}
For any subset $J$ of $\mathbb{Z}_n=\{0,1,\cdots,n-1\}$, a {\it hole} $H$ of $J$ is defined as a nonempty set such that $H=\{a_0+1,a_0+2,\cdots,a_0+|H|\}\subseteq \mathbb{Z}_n\backslash J$ and $a_0,a_{0}+|H|+1\in J$. And we denote the set of all the holes of $J$ by $\mathbb{H}(J)$.
\end{Definition}

We say $J$ is a {\it successive subset} of $\mathbb{Z}_n$, if $|\mathbb{H}(J)|\leq 1$.




\section{Generalized $b$-symbol weights of linear codes}
In this section, we give some general properties of generalized  $b$-symbol weights of linear codes. The following lemma gives a description on the relationship between the (Hamming) 1-symbol weight $w_{1}(D)$ and the $b$-symbol weight $w_{b}(D)$ for any $\mathbb{F}_q$-subspace $D$ of $\mathbb{F}_q^{n}$.

\begin{lem} \label{relationship-1}
Assume the notations given above. Let $D$ be an $\mathbb{F}_q$-subspace of $\mathbb{F}_q^{n}$. Then
$$w_{b}(D)=w_{1}(D)+\sum_{H\in \mathbb{H}(\chi_1(D)),\,|H|\leq b-1}|H|+\sum_{H\in \mathbb{H}(\chi_1(D)),\,|H|\ge b}(b-1),$$ where $\mathbb{H}(\chi_1(D))$ is the set of all the holes of $\chi_1(D)$.
\end{lem}

\begin{proof}
If $i\in\chi_1(D)$, there exists $\mathbf{x}=(x_{0},x_1,\cdots,x_{n-1})\in D$ such that $x_{i}\neq 0$. Then we know
$$(x_{i-b+1},x_{i-b+2},\cdots,x_{i}),(x_{i-b+2},x_{i-b+3},\cdots,x_{i+1}),\cdots,(x_{i},x_{i+1},\cdots,x_{i+b-1})$$ are not $\mathbf{0}$. Hence $i-b+1,i-b+2,\cdots,i\in \chi_b(D)$.

If $H=\{a_0+1,a_0+2,\cdots,a_0+|H|\}$ is an element of $\mathbb{H}(\chi_1(D))$ and $|H|\leq b-1$, we have $H\subseteq\chi_b(D)$ since $$a_0+|H|+1\in \chi_1(D).$$

If $H=\{a_0+1,a_0+2,\cdots,a_0+|H|\}$ is an element of $\mathbb{H}(\chi_1(D))$ and $|H|\ge b$, we have $$\{a_0+|H|-b+2,a_0+|H|-b+3,\cdots,a_0+|H|\}\subseteq\chi_b(D)$$ and $\{a_0+1,a_0+2,\cdots,a_0+|H|-b+1\}\subseteq\mathbb{Z}\backslash \chi_b(D)$.
 Hence $$w_{b}(D)=w_{1}(D)+\sum_{H\in \mathbb{H}(\chi_1(D)),\,|H|\leq b-1}|H|+\sum_{H\in \mathbb{H}(\chi_1(D)),\,|H|\ge b}(b-1).$$
\end{proof}

\begin{Theorem}\label{4.3}
Assume the notations given above. Let $C$ be an $[n,k]$-linear code over $\mathbb{F}_q$. For $1\leq r\leq k-1$, we have
$$\min\{d_{1}^{\,r}(C)+b-1,\,n\}\leq d_{b}^{\,r}(C)\leq \min\{bd_{1}^{\,r}(C),\,n\}$$

\end{Theorem}

\begin{proof}
If $d_{1}^{\,r}(C)+b-1\ge n$, then $$|\chi_1(D)|\ge d_{1}^{\,r}(C)\ge n-b+1$$
and $n-|\chi_1(D)|\leq b-1$ for any $\mathbb{F}_q$-subspace $D$ of dimension $r$ of $C$.
Then we have $|H|\leq b-1$ for any $H\in \mathbb{H}(\chi_1(D))$.
 By Lemma~\ref{relationship-1},$$w_{b}(D)=w_{1}(D)+\sum_{H\in \mathbb{H}(\chi_1(D)),\,|H|\leq b-1}|H|=n$$ for any $\mathbb{F}_q$-subspace $D$ of dimension $r$ of $C$. Hence $\min\{d_{1}^{\,r}(C)+b-1,\,n\}\leq d_{b}^{\,r}(C)=n$.

 If $d_{1}^{\,r}(C)+b-1<n$.
There exists an $\mathbb{F}_q$-subspace $E$ of $C$ such that $\dim(E)=r$ and $w_{b}(E)=d_{b}^{\,r}(C)$.
By Lemma~\ref{relationship-1}, we have $$w_{b}(E)=w_{1}(E)+\sum_{H\in \mathbb{H}(\chi_1(E)),\,|H|\leq b-1}|H|+\sum_{H\in \mathbb{H}(\chi_1(E)),\,|H|\ge b}(b-1).$$

If $|H|\leq b-1$ for any $H\in \mathbb{H}(\chi_1(E))$, then $d_{b}^{\,r}(C)=w_{b}(E) =n>d_{1}^{\,r}(C)+b-1$.

If there exists $H\in \mathbb{H}(\chi_1(E))$ such that $|H|\ge b$, then $$d_{b}^{\,r}(C)=w_{b}(E) \ge w_{1}(E)+\sum_{H\in \mathbb{H}(\chi_1(E)),\,|H|\ge b}(b-1)\ge w_{1}(E)+b-1 \ge d_{1}^{\,r}(C)+b-1.$$
Hence, we get $\min\{d_{1}^{\,r}(C)+b-1,\,n\}\leq d_{b}^{\,r}(C)$.

Now we prove that $d_{b}^{\,r}(C)\leq bd_{1}^{\,r}(C)$.
There exists an $\mathbb{F}_q$-subspace $D$ of $C$ such that $\dim(D)=r$ and $w_{1}(D)=d_{1}^{\,r}(C)$.
 By Lemma~\ref{relationship-1}, we have
 \begin{align*}
   w_{b}(D) & =w_{1}(D)+\sum_{H\in \mathbb{H}(\chi_1(D)),\,|H|\leq b-1}|H|+\sum_{H\in \mathbb{H}(\chi_1(D)),\,|H|\ge b}(b-1) \\
    & \leq w_{1}(D)+(b-1)|\mathbb{H}(\chi_1(D))|\\
    &\leq bw_{1}(D).
 \end{align*}
Hence $ d_{b}^{\,r}(C)\leq w_{b}(D)\leq bw_{1}(D)=bd_{1}^{\,r}(C)$.

\end{proof}

\begin{Remark}
When $r=1$, the statement (b) of Theorem~\ref{4.3} was proved in Proposition 3 of \cite{YBS}.
When $r=2$, the statement (b) of Theorem~\ref{4.3} was proved in Theorem 3.2 of \cite{LP}.
\end{Remark}

For any $[n,k]$-linear code $C$ over $\mathbb{F}_q$, it is easy to know that $$b\leq d_{b}^{\,1}(C)\leq d_{b}^{\,2}(C)\leq \cdots\leq d_{b}^{\,k-1}(C) \leq d_{b}^{\,k}(C)\leq n.$$ We give an improvement of this inequalities in the next theorem.

\begin{Theorem}\label{pair monotonicity}
Assume the notations given above. Then
\begin{description}
             \item[(a)] For $1\leq r\leq k-1$, if $d_{b}^{\,r+1}(C)<n$, then $d_{b}^{\,r}(C)< d_{b}^{\,r+1}(C)$.
             \item[(b)] If $k>b$, then $b\leq d_{b}^{\,1}(C)< d_{b}^{\,2}(C)< \cdots< d_{b}^{\,k-b}(C)< d_{b}^{\,k-b+1}(C)\leq\cdots \leq d_{b}^{\,k}(C)\leq n$.
             \item[(c) (Singleton-like bound for generalized $b$-symbol weights)] For $1\le r\le k-1$,  $$d_{b}^{\,r}(C)\leq \min\{n-k+b+r-1,n\}.$$
\noindent In particular, when $r=1$, $d_{b}^{\,1}(C)\leq \min\{n-k+b,n\}$, which  is called the Singleton-like bound for $b$-symbol weights.
\end{description}

\end{Theorem}

\begin{proof}

(a) There exists an $\mathbb{F}_q$-subspace $E$ of $C$ such that $\dim(E)=r+1$ and $w_{b}(E)=d_{b}^{\,r+1}(C)$. Then there exists $$H=\{a_0+1,a_0+2,\cdots,a_0+|H|\}\in \mathbb{H}(\chi_1(E))$$ such that $|H|\ge b$ otherwise $n=w_{b}(E)=d_{b}^{\,r+1}(C)$ by Lemma~\ref{relationship-1}. Suppose $$\tilde{E}=\{\mathbf{x}=(x_{0},x_1,\cdots,x_{n-1})\in E\,|\,x_{a_0}=0\}.$$
Then we get $\tilde{E}<E$ and $\dim(\tilde{E})=r$ since $a_0\in \chi_1(E)$.
Then $\chi_1(\tilde{E})\subseteq\chi_1(E)\setminus\{a_0\}$ and $a_0\in \chi_b(E)\setminus\chi_b(\tilde{E})$ since $\{a_0+1,a_0+2,\cdots,a_0+|H|\}\subseteq \mathbb{Z}_n\setminus \chi_1(E)$ and $|H|\ge b$.
Hence $d_{b}^{\,r}(C)\leq w_b(\tilde{E})<w_b(E)=d_{b}^{\,r+1}(C)$.

(b) First we prove $d_{b}^{\,k-b}(C)<n$. Let $\mathfrak{H}(C)=\max\{\,|H|\,|\,H\in \mathbb{H}(\chi_1(C))\}$, we assume $\mathfrak{H}(C)=0$  when $ \mathbb{H}(\chi_1(C))=\emptyset$.

If $\mathfrak{H}(C)\ge b$, then $ d_{b}^{\,k-b}(C)\leq d_{b}^{\,k}(C)<n$ by the definition of $d_{b}^{\,k}(C)$ and Lemma~\ref{relationship-1}.

If $0\leq \mathfrak{H}(C)\leq b-1$, we claim $d_{b}^{\,k-(b-\mathfrak{H}(C))}(C)<n$. We prove that by induction on $t=b-\mathfrak{H}(C)$ where $1\leq t\leq b$.

 Suppose $t=1$ and $\mathfrak{H}(C)=b-1$. If $ \mathbb{H}(\chi_1(C))$ is not empty which means $b\ge 2$, then there exists an $H_1\in \mathbb{H}(\chi_1(C))$ such that $|H_1|=b-1$ and $$H_1=\{a_1+1,a_1+2,\cdots ,a_1+b-1\}.$$

 If $ \mathbb{H}(\chi_1(C))$ is empty which means $b=1$, we take any $a_1\in \mathbb{Z}_n$.
 Let $$C_1=\{\mathbf{c}=(c_0,c_1,\cdots,c_{n-1})\in C\,|\,c_{a_1}=0\}$$ whenever $ \mathbb{H}(\chi_1(C))$ is empty or not.
 Then $\chi_1(C_1)\subseteq \chi_1(C)\setminus\{a_1\}$ and $$\dim(C)=\dim(C_1)+1.$$
 Hence $a_1\in \chi_b(C)\setminus \chi_b(C_1)$ and $d_{b}^{\,k-1}(C)\leq w_b(C_1)<w_b(C)=d_{b}^{\,k}(C)\leq n$.

Now suppose $2\leq t \leq b$. If $ \mathbb{H}(\chi_1(C))$ is not empty which means $b\ge t+1$, then there exists an $H_2\in \mathbb{H}(\chi_1(C))$ such that $|H_2|=b-t$ and $$H_2=\{a_2+1,a_2+2,\cdots ,a_2+b-t\}. $$

If $ \mathbb{H}(\chi_1(C))$ is empty which means $b=t$, we take any $a_2\in \mathbb{Z}_n$.
Let $$C_2=\{\mathbf{c}=(c_0,c_1,\cdots,c_{n-1})\in C\,|\,c_{a_2}=0\}$$ whenever $ \mathbb{H}(\chi_1(C))$ is empty or not.
Then $$\chi_1(C_2)\subseteq \chi_1(C)\setminus\{a_2\}$$ and $\dim(C)=\dim(C_2)+1$.
Hence $\mathfrak{H}(C_2)\ge\mathfrak{H}(C)+1$.

If $\mathfrak{H}(C_2)\ge b$, then $$d_{b}^{\,k-(b-\mathfrak{H}(C))}(C)\leq d_{b}^{\,k-(b-\mathfrak{H}(C))}(C_2)<d_{b}^{\,k-1}(C_2)<n$$ where $b-\mathfrak{H}(C)=t\ge 2$.

If $\mathfrak{H}(C_2)\leq b-1$, then $1\leq b-\mathfrak{H}(C_2)<b-\mathfrak{H}(C)=t$ and $$d_{b}^{\,(k-1)-(b-\mathfrak{H}(C_2))}(C_2)<n.$$ by induction.
Therefore, we have  $$d_{b}^{\,k-(b-\mathfrak{H}(C))}(C)\leq d_{b}^{\,k-b-1+\mathfrak{H}(C_2)}(C)\leq d_{b}^{\,k-b-1+\mathfrak{H}(C_2)}(C_2)=d_{b}^{\,(k-1)-(b-\mathfrak{H}(C_2))}(C_2)<n.$$
This implies that  $d_{b}^{\,k-b}(C)\leq d_{b}^{\,k-(b-\mathfrak{H}(C))}(C)<n$ when $0\leq \mathfrak{H}(C)\leq b-1$. Hence $d_{b}^{\,k-b}(C)<n$ as we claimed.

By (a), we have $$b\leq d_{b}^{\,1}(C)< d_{b}^{\,2}(C)< \cdots< d_{b}^{\,k-b}(C)\leq d_{b}^{\,k-b+1}(C)\leq\cdots \leq d_{b}^{\,k}(C)\leq n.$$
Suppose $d_{b}^{\,k-b}(C)= d_{b}^{\,k-b+1}(C)<n$, then $d_{b}^{\,k-b}(C)<d_{b}^{\,k-b+1}(C)$, and  by (a) which is a contradiction. Hence $$b\leq d_{b}^{\,1}(C)< d_{b}^{\,2}(C)< \cdots< d_{b}^{\,k-b}(C)< d_{b}^{\,k-b+1}(C)\leq\cdots \leq d_{b}^{\,k}(C)\leq n.$$

(c) If $1\leq r \leq k-b$, then $$d_{b}^{\,r}(C)\leq d_{b}^{\,r+1}(C)-1\leq \cdots\leq  d_{b}^{\,k-b}(C)-(k-b-r)\leq n-1-(k-b-r)=n+b+r-k-1\leq n.$$
Since $d_{b}^{\,r}(C)\leq n \leq n+b+r-k-1$ when $k-b+1\leq r\leq k,$ we have $$d_{b}^{\,r}(C)\leq\min\{n+b+r-k-1, n\}$$ for $1\leq r \leq k.$

\end{proof}

\begin{Remark}
When $b=1$, the statement (b) of Theorem~\ref{pair monotonicity} is $$1\leq d_{1}^{\,1}(C)<d_{1}^{\,2}(C)<\cdots<d_{1}^{\,k-1}(C) < d_{1}^{\,k}(C)\leq n$$ which was proved by Wei \cite{W}. When $b=2$, the statement (b) of Theorem~\ref{pair monotonicity} is $$2\leq d_{2}^{\,1}(C)<d_{2}^{\,2}(C)<\cdots<d_{2}^{\,k-1}(C) \leq d_{2}^{\,k}(C)\leq n$$ which was proved by Liu and Pan \cite{LP}.
\end{Remark}

From Theorem \ref{pair monotonicity}, we have the following definition.
\begin{Definition}\label{bmds}
An $[n,k]$-linear code $C$ over $\mathbb{F}_q$ with $d_{b}(C)= \min\{n-k+b,n\}$ is called a $b$-symbol maximum distance separable ($\,b$-symbol MDS) code.
\end{Definition}

\begin{Remark}
Theorem 2.4 of \cite{DZ} gives $d_{b}^{\,1}(C)\leq n-k+b$, but there is no $[n,k]$-linear code $C$ over $\mathbb{F}_q$ such that $d_{b}^{\,1}(C)= n-k+b$ when $k<b$. Hence our bound in statement (d) of Theorem~\ref{pair monotonicity} is an improvement of Theorem 2.4 of \cite{DZ}. Also there is an $[n,k]$-linear code $C$ over $\mathbb{F}_q$ such that $d_{b}^{\,1}(C)= \min\{n+b-k,n\} $ when $k<b$, for example $1$-MDS codes by using Theorem~\ref{GBW}.
\end{Remark}



For any subset $J\subseteq \mathbb{Z}_n$, let $J[b]=\cup_{i=0}^{b-1}(J+i)$ and $J[-b]=\cup_{i=0}^{b-1}(J-i)$.

\begin{lem}\label{n8.1}
For any $\mathbb{F}_q$-subspace $D$ of $\mathbb{F}_q^n$, we have
\begin{description}
  \item[(a)] $j\in\mathbb{Z}_n\setminus \chi_b(D)$ if and only if $\{j\}[b]\subseteq\mathbb{Z}_n\setminus \chi_1(D).$
\item[(b)] $\chi_1(D)[-b]=\chi_b(D)$.
\item[(c)] For any $1\leq b\leq n-1$, $\chi_{b}(D)[-1]=\chi_{b+1}(D)$.
\end{description}
\end{lem}

\begin{proof}
(a) It is easy to prove that by definitions of $\chi_1(D)$ and $\chi_b(D)$.

(b) For any $j\in \chi_1(D)[-b]$, there exists $0\leq i \leq b-1$ such that $j\in \chi_1(D)-i$. Then $j+i\in \chi_1(D)$ and $j\in \chi_b(D)$ by (a). Hence $\chi_1(D)[-b]\subseteq\chi_b(D)$.

 For any $j_1\in \chi_b(D)$, there exists $\mathbf{x}=(x_0,x_1,\cdots,x_{n-1})\in D$ such that $$(x_{j_1},x_{j_1+1},\cdots,x_{j_1+b-1})\neq \mathbf{0}.$$
Then there exists $0\leq i_1 \leq b-1$ such that $x_{j_1+i_1}\neq 0$ and $j_1+i_1\in \chi_1(D)$.
Hence $j_1\in \chi_1(D)[-b]$ and $\chi_1(D)[-b]\supseteq\chi_b(D)$.

(c) By (b), we know
\begin{align*}
                  \chi_{b}(D)[-1] & =\chi_{b}(D)\bigcup (\chi_{b}(D)-1) \\ &=\chi_1(D)[-b]\bigcup(\chi_1(D)[-b]-1)\\
                   &=\bigcup_{i=0}^{b-1}(\chi_1(D)-i)\bigcup(\bigcup_{i=0}^{b-1}(\chi_1(D)-i)-1)\\ &=\bigcup_{i=0}^{b-1}(\chi_1(D)-i)\bigcup(\bigcup_{i=1}^{b}(\chi_1(D)-i))\\
                   &=\bigcup_{i=0}^{b}(\chi_1(D)-i)=\chi_{b+1}(D).
\end{align*}
\end{proof}
For any $[n,k]$-linear code $C$ over $\mathbb{F}_q$, it is easy to know that $$1\leq d_{1}^{\,r}(C)\leq d_{2}^{\,r}(C)\leq \cdots\leq d_{n}^{\,r}(C)= n.$$ Then we give an improvement of this inequalities in the next theorem.

\begin{Theorem}\label{pair monotonicity  1}
Let $C$ be an $[n,k]$-linear code over $\mathbb{F}_q$. Let $1\leq r \leq k$. Then
\begin{description}
             \item[(a)] For $1\leq b\leq n-1$, if $d_{b+1}^{\,r}(C)<n$, then $d_{b}^{\,r}(C)< d_{b+1}^{\,r}(C)$.
             \item[(b)] $1\leq d_{1}^{\,r}(C)<\cdots<d_{k-r}^{\,r}(C)\leq d_{k-r+1}^{\,r}(C)\leq \cdots\leq d_{n}^{\,r}(C)= n.$
             \item[(c)] For $1\leq b\leq n-1$, $d_{b+1}^{\,r}(C)=d_{b}^{\,r}(C)+1$ if and only if there exists an $\mathbb{F}_q$-subspace $E$ of $C$ such that $\dim(E)=r$, $d_{b}^{\,r}(C)=w_{b}(E)<n$ and $\chi_{b}(E)$ is successive.
             \item[(d)] If there exists an $\mathbb{F}_q$-subspace $E$ of $C$ such that $\dim(E)=r$, $d_{1}^{\,r}(C)=w_{1}(E)$ and $\chi_{1}(E)$ is successive, then $d_{b}^{\,r}(C)=\min\{d_{1}^{\,r}(C)+b-1,n\}$.
\end{description}
\end{Theorem}

\begin{proof}
(a) For $1\leq b\leq n-1$, there exists an $\mathbb{F}_q$-subspace $E$ of $C$ such that $\dim(E)=r$ and $w_{b+1}(E)=d_{b+1}^{\,r}(C)$. By Lemma~\ref{relationship-1}, we know

\begin{align}\label{4.5.1}
  w_{b}(E) & =w_{1}(E)+\sum_{H\in\mathbb{H}(X),\,|H|\leq b-1}|H|+\sum_{H\in\mathbb{H}(X),\,|H|\ge b}(b-1)\notag\\
   & =w_{1}(E)+\sum_{H\in\mathbb{H}(X),\,|H|\leq b-1}|H|+\sum_{H\in\mathbb{H}(X),\,|H|= b}(b-1)+\sum_{H\in\mathbb{H}(X),\,|H|\ge b+1}(b-1)
\end{align}
and
\begin{align}\label{4.5.2}
  w_{b+1}(E) & =w_{1}(E)+\sum_{H\in\mathbb{H}(X),\,|H|\leq b}|H|+\sum_{H\in\mathbb{H}(X),\,|H|\ge b+1}b \notag\\
   & =w_{1}(E)+\sum_{H\in\mathbb{H}(X),\,|H|\leq b-1}|H|+\sum_{H\in\mathbb{H}(X),\,|H|=b}|H|+\sum_{H\in\mathbb{H}(X),\,|H|\ge b+1}b\notag\\
   & =w_{1}(E)+\sum_{H\in\mathbb{H}(X),\,|H|\leq b-1}|H|+\sum_{H\in\mathbb{H}(X),\,|H|=b}b+\sum_{H\in\mathbb{H}(X),\,|H|\ge b+1}b\\
   &\ge w_{b}(E) \notag,
\end{align}
where $X=\chi_1(E)$. If $d_{b+1}^{\,r}(C)<n$, then there exists $H\in \mathbb{H}(X)$ and $|H|\ge b+1$ otherwise $n=w_{b+1}(E)=d_{b+1}^{\,r}(C)$ by Lemma~\ref{relationship-1}. Then $$d_{b+1}^{\,r}(C)= w_{b+1}(E)> w_{b}(E)\ge d_{b}^{\,r}(C)$$ by Equation~\ref{4.5.1} and Equation~\ref{4.5.2}.

(b) By Theorem~\ref{pair monotonicity} (b), we know that $d_{k-r}^{\,r}(C)<n$. By (a), we get $$d_{1}^{\,r}(C)<d_{2}^{\,r}(C)<\cdots<d_{k-r}^{\,r}(C).$$

(c) Suppose $d_{b}^{\,r}(C)+1=d_{b+1}^{\,r}(C)$.

If $d_{b}^{\,r}(C)+1=d_{b+1}^{\,r}(C)=n$, then there exists an $\mathbb{F}_q$-subspace $E$ of $C$ such that $\dim(E)=r$ and $d_{b}^{\,r}(C)=w_{b}(E)=n-1$. Hence $\chi_{b}(E)$ is successive.

If $d_{b}^{\,r}(C)+1=d_{b+1}^{\,r}(C)<n$, then there exists an $\mathbb{F}_q$-subspace $E_1$ of $C$ such that $\dim(E_1)=r$ and $d_{b+1}^{\,r}(C)=w_{b+1}(E_1)<n$. We have
$$d_{b}^{\,r}(C)\leq w_{b}(E_1)<w_{b+1}(E_1)=d_{b+1}^{\,r}(C)$$ and
\begin{equation}\label{3.9a}
 w_{b}(E_1)+1=w_{b+1}(E_1)
\end{equation}
since $w_{b}(E_1)\leq w_{b+1}(E_1)<n$ and Lemma~\ref{n8.1} (c).
Hence $d_{b}^{\,r}(C)= w_{b}(E_1)$ and $\chi_{b}(E_1)$ is successive by Equation~\ref{3.9a} and Lemma~\ref{n8.1} (c).

Suppose there exists an $\mathbb{F}_q$-subspace $E_2$ of $C$ such that $\dim(E_2)=r$, $$d_{b}^{\,r}(C)=w_{b}(E_2)<n$$ and $\chi_{b}(E_2)$ is successive.
 By Lemma~\ref{n8.1} (c), we know
$$d_{b+1}^{\,r}(C)\leq w_{b+1}(E_2)=w_{b}(E_2)+1=d_{b}^{\,r}(C)+1$$
and $d_{b}^{\,r}(C)+1=d_{b+1}^{\,r}(C)$ by (a).

(d) We prove (d) by induction.
It is a trivial case when $b=1$.

Now suppose $2\leq b\leq n$, we have $$d_{b-1}^{\,r}(C)=\min\{d_{1}^{\,r}(C)+b-2,n\}$$ by induction.
If $d_{b-1}^{\,r}(C)=n$, then $d_{b}^{\,r}(C)=n=\min\{d_{1}^{\,r}(C)+b-1,n\}$.

If $d_{b-1}^{\,r}(C)=d_{1}^{\,r}(C)+b-2<n$, we know $\chi_{b-1}(E)$ is successive and $$w_{b-1}(E)=w_1(E)+b-2=d_{1}^{\,r}(C)+b-2=d_{b-1}^{\,r}(C)<n$$ since $\chi_{1}(E)$ is successive and $\chi_{b-1}(E)=\chi_{1}(E)[-(b-1)]$ by Lemma~\ref{n8.1} (c). By (c), we get $$d_{b}^{\,r}(C)=d_{b-1}^{\,r}(C)+1=\min\{d_{1}^{\,r}(C)+b-1,n\}.$$
\end{proof}

For two real number $n\times k$ matrixes $A=(a_{ij})_{n\times k}$ and $B=(b_{ij})_{n\times k}$, we assume $A\leq B$ when $a_{ij}\leq b_{ij}$ for any $1\leq i\leq n$ and $1\leq j\leq k$. Let
  $$D(n,k)=\left(\begin{array}{ccccc}
                       n-k+1 & n-k+2&\cdots&n-1&n\\
                       n-k+2 & n-k+3&\cdots&n&n\\
                       \vdots&\ddots&\vdots&\vdots&\vdots\\
                        n&n&\cdots&n&n\\
                        n&n&\cdots&n&n\\
                        \vdots&\vdots&\vdots&\vdots&\vdots\\
                        n&n&\cdots&n&n
\end{array}\right)_{n\times k} $$ for $ k\leq n \in \mathbb{N}^{+}$,
then we have the following theorem.

\begin{Theorem}\label{GBW}
Assume the notations given above. For any $[n,k]$-linear code $C$ over $\mathbb{F}_q$,
\begin{description}
  \item[(a)] Every row of the matrix $D(C)$ is increasing from left to right. And every column of the matrix $D(C)$ is increasing from up to down.
  \item[(b) (Singleton-like bound for the generalized weight matrix)] $D(C)\leq D(n,k)$.
  \item[(c)] For $1\leq b \leq n$, $C$ is a $b$-symbol MDS if and only if the $(b,1)$-element of $D(C)$ is same as $(b,1)$-element of $D(n,k)$ if and only if the $b$th row of $D(C)$ is same as $b$th row of $D(n,k)$.
   \item[(d)]  Let $b_0=\min\{1\leq b \leq n\,|\,C$ is a $b$-symbol MDS code $\}$, then $C$ is a $b$-symbol MDS for any $b_0\leq b \leq n$.
   \item[(e)]  In particular, $C$ is a $1$-symbol MDS if and only if $D(C)=D(n,k)$.
\end{description}
\end{Theorem}

\begin{proof}
(a) It is easy to prove.



Statements (b) and (c) have been proved in Theorem~\ref{pair monotonicity} (c).

(d) First $\{1\leq b \leq n\,|\,C$ is a $b$-symbol MDS code$\}$ is not empty set since $C$ is an $n$-symbol MDS code. Then we only need to prove (d) when $b_0<n$ and $b=b_0+1$.

If $d_{b}^{\,1}(C)=n$, then $C$ is a $b$-symbol MDS. If $d_{b}^{\,1}(C)<n$, then $$n+b_0-k=d_{b_0}^{\,1}(C)<d_{b}^{\,1}(C)\leq n+b-k$$ by Theorem~\ref{pair monotonicity  1} (b). By $b=b_0+1$, we have $d_{b}^{\,1}(C)= n+b-k$ and $C$ is a $b$-symbol MDS code.

(e) It is easy to prove by (c) and (d).
\end{proof}

\section{Generalized weight matrices of two classes of codes}
In this section, we calculate the generalized weight matrix $D(C)$ defined in Section 2, when $C$ is simplex codes or two especial Hamming codes.
 First we assume $$\mathbb{F}_q=\{\alpha_{0}=0,\alpha_1=1,\alpha_{2},\cdots,\alpha_{q-1}\}$$ and give an order on $\mathbb{F}_q$ which is $$\alpha_{0}\leq \alpha_1\leq \alpha_{2}\leq \cdots\leq \alpha_{q-1}.$$

Let  $\mathbf{x},\, \mathbf{y}\in \mathbb{F}_q^k$, we define an order on $\mathbb{F}_q^k$ by using lexicographical order as follows: Two vectors  $\mathbf{x},\, \mathbf{y}$ are called ordered, denoted by  $\mathbf{x}\leq \mathbf{y}$,  if and only if $$\mathbf{x}=(x_0,x_1,\cdots,x_{k-1}),\,\,\mathbf{y}=(y_0,y_1,\cdots,y_{k-1})$$ such that there exists $0\leq i_0\leq k-1$ such that $x_j=y_j$ for any $0\leq j\leq i_0-1$ and $x_{i_0} < y_{i_0}$ ( which means $x_{i_0} \neq y_{i_0}$ and $x_{i_0} \leq y_{i_0}$ ). And $\mathbf{x}< \mathbf{y}$ means $\mathbf{x}\leq \mathbf{y}$ and $\mathbf{x}\neq \mathbf{y}$.

  Recall ${\rm {\rm PG}}^{1}(\mathbb{F}_q^{k})=\{V_{1}^1,V_{2}^1,\cdots,V_{n_{1,k}}^1 \}$ be the set of all subspaces of dimension $1$ of $\mathbb{F}_q^k$.  Then there exists a unique $\mathbf{v}_{i}\in V_{i}$ such that the first non zero component of $\mathbf{v}_{i}$ is $1$ for any $1\leq i\leq n_{1,k}$.

  For $k\ge 1$, let $H_{q,k}=(\mathbf{x}^{T}_{1},\mathbf{x}^{T}_{2},\cdots,\mathbf{x}^{T}_{n_{1,k}})$ be the $k\times n_{1,k}$ matrix over $\mathbb{F}_q$ such that $\mathbf{x}_{i}\in \{\mathbf{v}_{i}\,|\,1\leq i \leq n_{1,k}\}$ for $1\leq i\leq n_{1,k}$ and $\mathbf{x}_{i}<\mathbf{x}_{i+1}$ for $1\leq i\leq n_{1,k}-1$, then the linear code over $\mathbb{F}_q $ with the generator matrix $H_{q,k}$ is called {\it simplex code} donated by $\textbf{S}_{q,k}$ and the linear code over $\mathbb{F}_q $ with the parity check matrix $H_{q,k}$ is called {\it Hamming code} donated by $\textbf{H}_{q,k}$.

\begin{lem} \label{dual Hamming}[Theorem 3 of \cite{W}]
Let $C$ be an $[n,k]$-linear code over $\mathbb{F}_q$. Then
$$\{d_{1}^{\,j}(C)\,|\,1\leq j\leq k\}=\{1,2,\cdots,n\}\setminus\{n+1-d_{1}^{\,j}(C^{\perp})\,|\,1\leq j\leq n-k\}.$$
\end{lem}


 For $k\ge 1$, let $F_{q,k}=(\mathbf{x}^{T}_{1},\mathbf{x}^{T}_{2},\cdots,\mathbf{x}^{T}_{q^k})$ be the $k\times q^k$ matrix over $\mathbb{F}_q$ such that $\mathbf{x}_{i}\in \mathbb{F}_q^k$ for $1\leq i\leq n_{1,k}$ and $\mathbf{x}_{i}<\mathbf{x}_{i+1}$ for $1\leq i\leq q^k$. By using this notion, we have the following theorem.

%

%




\begin{Theorem}\label{6.33}
Assume the notations given above.
\begin{description}
  \item[(a)] $H_{q,k}=\left(\begin{array}{cc}
                      0 &1\\

                        H_{q,k-1} &F_{q,k-1}
\end{array}\right) $ for $k\ge 2$.

  \item[(b)] Let $\textbf{S}_{q,k}$ be the simplex code over $\mathbb{F}_q$ with the generator matrix $H_{q,k}$ for $k\ge 1$, then
  $$d_{i}^{\,j}(\textbf{S}_{q,k})=\min\{\frac{q^k-q^{k-j}}{q-1}+i-1,n\}$$ for $1\leq j\leq k$ and $1\leq i\leq n$.

\end{description}
\end{Theorem}

\begin{proof}
(a) It is easy to prove by the definition of $H_{q,k}$.

(b) Since $d_{1}^{\,1}(\textbf{S}_{q,k})=w_1(\mathbf{c})$ for any nonzero $\mathbf{c}\in \textbf{S}_{q,k}$, we know that $$d_{1}^{\,j}(\textbf{S}_{q,k})=w_1(E)$$ for any $E\in {\rm PG}^{j}(\textbf{S}_{q,k})$ by Lemma 1 of \cite{FL}. Let $V^j$ be the $\mathbb{F}_q$-subspace of $C$ generated by first $j$ rows of the matrix $H_{q,k}$, then
$$d_{1}^{\,j}(\textbf{S}_{q,k})=w_1(V^j)=q^{k-1}+q^{k-2}+\cdots+q^{k-j})=\frac{q^k-q^{k-j}}{q-1}$$ for $1\leq j\leq k$ by (a). By Theorem~\ref{pair monotonicity  1} (d), we have $$d_{i}^{\,j}(\textbf{S}_{q,k})=\min\{\frac{q^k-q^{k-j}}{q-1}+i-1,n\}$$ for $1\leq j\leq k$ and $1\leq i\leq n$.

\end{proof}

\begin{Example}
Let $$H_{2,3}=\left(\begin{array}{ccccccc}
                       0&0&0&1&1&1&1\\
                       0&1&1&0&0&1&1\\
                       1&0&1&0&1&0&1
\end{array}\right)_{7\times 3}$$ and let $C_{2,3}$ be the $[7,4]$-Hamming code over $\mathbb{F}_2$ with the parity check matrix $H_{2,3}$.
Then $$G_{2,3}=\left(\begin{array}{ccccccc}
                       1&1&1&0&0&0&0\\
                       1&0&0&1&1&0&0\\
                       1&0&0&0&0&1&1\\
                       0&1&0&1&0&1&0
\end{array}\right)_{7\times 4}$$ is a generator matrix of $C_{2,3}$.

Then the first row of $D(C_{2,3}^{\bot})$ is $(4\,6\,7)$ by Theorem \ref{6.33} and the first row of $D(C_{2,3})$ is $(3\,5\,6\,7)$ by Corollary 3 and Corollary 4 of \cite{W}.
By Theorem~\ref{pair monotonicity  1} (a) (d) and Theorem~\ref{GBW} (a), we have $$D(C_{2,3}^{\bot})=\left(\begin{array}{ccc}
                       4&6&7\\
                       5&7&7\\
                       6&7&7\\
                       7&7&7\\
                       7&7&7\\
                       7&7&7\\
                       7&7&7
\end{array}\right)_{3\times 7},\,\,D(C_{2,3})=\left(\begin{array}{cccc}
                       3&5&6&7\\
                       4&6&7&7\\
                       5&7&7&7\\
                       6&7&7&7\\
                       7&7&7&7\\
                       7&7&7&7\\
                       7&7&7&7
\end{array}\right)_{4\times 7}   .$$
\end{Example}

\begin{Example}
Let $$H_{2,4}=\left(\begin{array}{ccccccccccccccc}
                       0&0&0&0&0&0&0&1&1&1&1&1&1&1&1\\
                       0&0&0&1&1&1&1&0&0&0&0&1&1&1&1\\
                       0&1&1&0&0&1&1&0&0&1&1&0&0&1&1\\
                       1&0&1&0&1&0&1&0&1&0&1&0&1&0&1
\end{array}\right)_{4\times 15} .$$ We denote the $i$th row vector of $H_{2,4}$ by $\alpha_i$ for $1\leq i\leq 4$.

Let $$A=\left(\begin{array}{ccccccccccccccc}
                       1&1&1&0&0&0&0&0&0&0&0&0&0&0&0\\
                       1&0&0&1&1&0&0&0&0&0&0&0&0&0&0\\
                       1&0&0&0&0&1&1&0&0&0&0&0&0&0&0\\
                       1&0&0&0&0&0&0&1&1&0&0&0&0&0&0\\
                       1&0&0&0&0&0&0&0&0&1&1&0&0&0&0\\
                       1&0&0&0&0&0&0&0&0&0&0&1&1&0&0\\
                       1&0&0&0&0&0&0&0&0&0&0&0&0&1&1\\
                       0&1&0&1&0&1&0&0&0&0&0&0&0&0&0\\
                       0&1&0&0&1&0&1&0&0&0&0&0&0&0&0\\
                       0&1&0&0&0&0&0&1&0&1&0&0&0&0&0\\
                       0&1&0&0&0&0&0&0&1&0&1&0&0&0&0\\
                       0&1&0&0&0&0&0&0&0&0&0&1&0&1&0\\
                       0&1&0&0&0&0&0&0&0&0&0&0&1&0&1
\end{array}\right)_{13\times 15}$$ and we denote the $i$th row vector of $A$ by $\beta_i$ for $1\leq i\leq 13$.
Let $C_{2,4}$ be the $[15,11]$-Hamming code over $\mathbb{F}_2$ with the parity check matrix $H_{2,4}$, then $\beta_i\in C_{2,4}$ for $1\leq i\leq 13$. Then the first row of $D(C_{2,4}^{\bot})$ is $$(8\,12\,14\,15)$$ by Theorem \ref{6.33}, and the first row of $D(C_{2,4})$ is $$(3\,5\,6\,7\,9\,10\,11\,12\,13\,14\,15)$$ by Corollaries 3 and 4 of \cite{W}.
By Theorems~\ref{pair monotonicity} (a), \ref{pair monotonicity  1} (a) and (d), and Theorem~\ref{GBW} (a), we have $$D(C_{2,4}^{\bot})=\left(\begin{array}{cccc}
                       8&12&14&15\\
                       9&13&15&15\\
                       10&14&15&15\\
                       11&15&15&15\\
                       12&15&15&15\\
                       13&15&15&15\\
                       14&15&15&15\\
                       15&15&15&15\\
                       15&15&15&15\\
                       15&15&15&15\\
                       15&15&15&15\\
                       15&15&15&15\\
                       15&15&15&15\\
                       15&15&15&15\\
                       15&15&15&15
\end{array}\right)_{15\times 4},\,\,D(C_{2,4})=\left(\begin{array}{ccccccccccc}
                       3&5&6&7&9&10&11&12&13&14&15\\
                       4&6&7&8&10&11&12&13&14&15&15\\
                       5&7&8&9&11&12&13&14&15&15&15\\
                       6&8&9&10&12&13&14&15&15&15&15\\
                       7&9&10&11&13&14&15&15&15&15&15\\
                       8&10&11&12&14&15&15&15&15&15&15\\
                       9&11&12&13&15&15&15&15&15&15&15\\
                       10&12&13&14&15&15&15&15&15&15&15\\
                       11&13&14&15&15&15&15&15&15&15&15\\
                       12&14&15&15&15&15&15&15&15&15&15\\
                       13&15&15&15&15&15&15&15&15&15&15\\
                       14&15&15&15&15&15&15&15&15&15&15\\
                       15&15&15&15&15&15&15&15&15&15&15\\
                       15&15&15&15&15&15&15&15&15&15&15\\
                       15&15&15&15&15&15&15&15&15&15&15

\end{array}\right)_{15\times 11}   .$$
In fact, we know that the first column of $D(C_{2,4}^{\bot})$ is $$(8\,9\,10\,11\,12\,13\,14\,15\,15\,15\,15\,15\,15\,15\,15)^T$$
since there exists an $\mathbb{F}_q$-subspace $E_1$ of $C_{2,4}^{\bot}$ such that $\dim(E_1)=1$, $d_{b}^{\,1}(C)=w_{b}(E_1)<n$ and $\chi_{1}(E_1)$ is successive, where $E_1=\langle\alpha_1\rangle$. Analogously, we can calculate the second column of $D(C_{2,4}^{\bot})$ by $E_2=\langle\alpha_1,\,\alpha_2\rangle$ and Theorem~\ref{pair monotonicity  1} (a) (d).

And analogously calculating the first column of $D(C_{2,4})$ is by $D_1=\langle\beta_1\rangle$.

Calculating the second column of $D(C_{2,4})$ is by $D_2=\langle\beta_1,\beta_2\rangle$.

Calculating the 3-th column of $D(C_{2,4})$ is by $D_3=\langle\beta_1,\beta_2,\beta_8\rangle$.

Calculating the 4-th column of $D(C_{2,4})$ is by $D_4=\langle\beta_1,\beta_2,\beta_3,\beta_8\rangle$.

And the calculation of the rest column of $D(C_{2,4})$ is by using Theorem~\ref{pair monotonicity} (a) and Theorem~\ref{pair monotonicity  1} (a).
\end{Example}

\begin{Remark}
By Corollary A.2 of \cite{LP}, it is easy to get the utility performance of simplex codes and Hamming codes we have calculated in symbol-pair read wire-tap channels of type II.
\end{Remark}

\section{$b$-symbol MDS codes}
In \cite{DZ}, the authors gave a sufficient condition for the existence of $b$-symbol MDS codes by using parity check matrices of linear codes. And they \cite{DZ}  constructed $b$-symbol MDS codes by using this condition. In this section, we give a necessary and sufficient condition for a linear code to be a $b$-symbol MDS code by using the generator matrix and  the parity check matrix of this linear code,  respectively.

Recall that  we have assumed $G=(G_{0},\cdots,G_{n-1})$ is a generator matrix of an $[n,k]$-linear code $C$ over $\mathbb{F}_q$.
Now we take all $i$th columns of $G$ such that $i\in J[b]$ and put them together to form a submatrix of $G$, which is denoted by $[G_j\,|\,j\in J[b]\,]$.

\begin{Theorem}\label{7.1}
Assume the notations given above. An $[n,k]$-linear code $C$ over $\mathbb{F}_q$ is a $b$-symbol MDS code if and only if $rank([G_j\,|\,j\in J[b]\,])=k$ for any $J\subseteq \mathbb{Z}_n$ such that $$|J|= \max\{k-b,0\}+1.$$
\end{Theorem}

\begin{proof}
It is enough to prove that an $[n,k]$-linear code $C$ over $\mathbb{F}_q$ is a $b$-symbol MDS code if and only if $rank([G_j\,|\,j\in J[b]\,])=k$ for any $J\subseteq \mathbb{Z}_n$ such that $|J|\ge \max\{k-b,0\}+1$.

Suppose $C$ is not a $b$-symbol MDS code. There exists a nonzero $\mathbf{c}_0\in C$ such that $$w_b(\mathbf{c}_0)\leq \min\{n-k+b,n\}-1$$ and a nonzero $\mathbf{y}_0\in \mathbb{F}_q^k$ such that $\mathbf{c}_0=\mathbf{y}_0G$. Let $J_0=\mathbb{Z}_n\setminus \chi_b(\mathbf{c}_0)$, then $$|J_0|=n-w_b(\mathbf{c}_0)\ge \max\{k-b,0\}+1.$$
By Lemma ~\ref{n8.1} (a), we have $J_0[b]\subseteq \mathbb{Z}_n\setminus \chi_1(\mathbf{c}_0)$ and $\mathbf{y}_0[G_j\,|\,j\in J_0[b]\,]=\mathbf{0}$. Hence $$rank([G_j\,|\,j\in J_0[b]\,])\leq k-1$$ which is a contradiction.

Assume there exists a subset $J_1\subseteq \mathbb{Z}_n$ such that $|J_1|\ge \max\{k-b,0\}+1$ and $$rank([G_j\,|\,j\in J_1[b]\,])\leq k-1.$$
Then there exists a nonzero $\mathbf{y}_1\in \mathbb{F}_q^k$ such that $\mathbf{y}_1[G_j\,|\,j\in J_1[b]\,]=\mathbf{0}$.
Assume $\mathbf{c}_1=\mathbf{y}_1G$, then $J_1[b]\subseteq \mathbb{Z}_n\setminus \chi_1(\mathbf{c}_1)$ and $J_1\subseteq \mathbb{Z}_n\setminus \chi_b(\mathbf{c}_1)$ by Lemma \ref{n8.1} (a). Hence $$n-w_b(\mathbf{c}_1)=|\mathbb{Z}_n\setminus \chi_b(\mathbf{c}_1)|\ge |J_1|\ge \max\{k-b,0\}+1$$ and $$w_b(\mathbf{c}_1)\leq \min\{n-k+b,n\}-1 .$$ That is a contradiction, since $C$ is a $b$-symbol MDS code.
\end{proof}

When $b=1$, we get the usually necessary and sufficient condition for a linear code to be a $1$-symbol MDS code (Hamming MDS code) by using generator matrixes.
\begin{Corollary}\label{5.3a}
Assume the notations given above. An $[n,k]$-linear code $C$ over $\mathbb{F}_q$ is a $1$-symbol MDS code if and only if $rank([G_j\,|\,j\in J\,])=k$ for any $J\subseteq \mathbb{Z}_n$ such that $|J|= k$.
\end{Corollary}

\begin{Corollary}\label{5.3b}
Assume the notations given above and $b\ge k$. An $[n,k]$-linear code $C$ over $\mathbb{F}_q$ is a $b$-symbol MDS code if and only if $rank([G_j\,|\,i\leq j\leq i+b-1 ])=k$ for any $0\leq i\leq n-1$.
\end{Corollary}

Given two positive integers $b$ and $k$ such that $b\ge k\ge 1$, we can construct an $[n,k]$-linear code $C$ over $\mathbb{F}_q$ such that $C$ is a $b$-symbol MDS code and $n$ is as large as possible in the following example.
\begin{Example}\label{4.4}
Given two positive integers $b$ and $k$ such that $b\ge k$, there exists a $k\times b$ matrix $\tilde{G}_1$ over $\mathbb{F}_q$ such that $rank(\tilde{G}_1)=k$. For any $t\in\mathbb{N}^{+}$, we construct a $[tb,k]$-linear code $C_t$ over $\mathbb{F}_q$ with a generator matrix $\tilde{G}_t=[\tilde{G}_1,\tilde{G}_1,\cdots,\tilde{G}_1]$, where $\tilde{G}_1$ repeats $t$ times in $\tilde{G}_{t}$. By Corollary~\ref{5.3b}, we know that the linear code $C_t$ is a $b$-symbol MDS code such that the $b$-symbol weight of any nonzero codeword is $n$.
\end{Example}

Given two positive integers $b$ and $k$ such that $1\leq b\leq  k-1$, we give a bound of length $n$ of $[n,k]$-linear codes which are $b$-symbol MDS codes in the following corollary.
\begin{Corollary}\label{5.3c}
Assume the notations given above and $1\leq b\leq k-1$. For any $[n,k]$-linear code $C$ over $\mathbb{F}_q$ which is a $b$-symbol MDS code, then $n\leq n_{1,k}$.
\end{Corollary}

\begin{proof}
We only need prove that $n\leq n_{1,k}$ for any $[n,k]$-linear code $C$ over $\mathbb{F}_q$ which is a $(k-1)$-symbol MDS code, since any $[n,k]$-linear code $C$ over $\mathbb{F}_q$ which is a $b$-symbol MDS code for $1\leq b\leq k-1$ is a $(k-1)$-symbol MDS code by Theorem~\ref{GBW} (d).

By Theorem~\ref{7.1}, we know an $[n,k]$-linear code $C$ over $\mathbb{F}_q$ is a $(k-1)$-symbol MDS code if and only if $rank([G_j\,|\,j\in J[k-1]\,])=k$ for any $J\subseteq \mathbb{Z}_n$ such that $|J|=2.$ For any $0\leq j \leq n-1$, there is a $\mathbb{F}_q$-subspace $V_j$ of $\mathbb{F}_q^k$ such that $\dim(V_j)=k-1$ and $$\{G_j,G_{j+1},\cdots,G_{j+k-2}\}\subseteq V_j.$$

Suppose $n>n_{1,k}=n_{k-1,k}$ which is the number of all $\mathbb{F}_q$-subspaces of dimension $k-1$ of $\mathbb{F}_q^k$, then there exists $j_1$ and $j_2$ such that $$0\leq j_1<j_2 \leq n-1$$ and $V_{j_1}=V_{j_2}$. Let $J_1=\{j_1,j_2\}$, then $rank([G_j\,|\,j\in J_1[k-1]\,])<k$ which is a contradiction.
\end{proof}


We assume that $H=(H_{0},\cdots,H_{n-1})$ is a parity check matrix of an $[n,k]$-linear code $C$ over $\mathbb{F}_q$. Then we take all $i$th column of $H$ such that $i\in J$ and put them together to form a submatrix of $H$, which is  denoted by $[H_j\,|\,j\in J\,]$.
\begin{Theorem}\label{7.2}
Assume the notations given above. An $[n,k]$-linear code $C$ over $\mathbb{F}_q$ is a $b$-symbol MDS code if and only if $rank([H_j\,|\,j\in J\,])=|J|$ for any $J\subseteq \mathbb{Z}_n$ such that $$|J[-b]|\leq \min\{n-k+b,n\}-1.$$
\end{Theorem}

\begin{proof}
(a) Suppose $C$ is not a $b$-symbol MDS code. There exist a nonzero $\mathbf{c}\in C$ such that $$w_b(\mathbf{c})\leq \min\{n-k+b,n\}-1.$$
Let $J_0=\chi_1(\mathbf{c})$. By Lemma ~\ref{n8.1} (b), we have $$|J_0[-b]|=|\chi_1(\mathbf{c})[-b]|=|\chi_b(\mathbf{c})|=w_b(\mathbf{c})\leq \min\{n-k+b,n\}-1.$$
Since $H\mathbf{c}^{T}=0$, we have $\sum_{j\in J_0} H_{j}c_j=0$ where $\mathbf{c}=(c_0,c_1,\cdots,c_{n-1})$. Then $$rank([H_j\,|\,j\in J_0\,])<|J_0|$$ which is a contradiction.

If there exists a $J_1\subseteq \mathbb{Z}_n$ such that $|J_1[-b]|\leq \min\{n-k+b,n\}-1$ and $$rank([H_j\,|\,j\in J_1\,])<|J_1|.$$
Then there exists a codeword $\mathbf{x}=(x_0,x_1,\cdots,x_{n-1})\in C$ such that $x_j=0$ for any $j\in \mathbb{Z}_n\backslash J_1$ and $$\sum_{j\in J_1} H_{j}x_j=0.$$
By Lemma ~\ref{n8.1} (b), we have $\chi_1(\mathbf{x})\subseteq J_1$ and $\chi_b(\mathbf{x})=\chi_1(\mathbf{x})[-b]\subseteq J_1[-b]$.
Hence $$w_b(\mathbf{x})=|\chi_b(\mathbf{x})|\leq |J_1[-b]|\leq \min\{n-k+b,n\}-1$$ that is a contradiction, since $C$ is a $b$-symbol MDS.
\end{proof}

When $b=1$, we get the usually necessary and sufficient condition for a linear code to be a $1$-symbol MDS code by using parity check matrixes.
\begin{Corollary}\label{5.5a}
Assume the notations given above. An $[n,k]$-linear code $C$ over $\mathbb{F}_q$ is a $1$-symbol MDS code if and only if $rank([H_j\,|\,j\in J\,])=n-k$ for any $J\subseteq \mathbb{Z}_n$ such that $|J|=n-k$.
\end{Corollary}

By Corollaries~\ref{5.3a} and \ref{5.5a}, an $[n,k]$-linear code $C$ over $\mathbb{F}_q$ is a $1$-symbol MDS code if and only if the dual $C^{\perp}$ of $C$ is a $1$-symbol MDS code. But the dual $C^{\perp}$ may not be a $b$-symbol MDS code, when $b\ge 2$ and $C$ is a $b$-symbol MDS code. And we know this by the following example.

\begin{Example}
Let $C$ be the linear code over $\mathbb{F}_2$ with a generator matrix $\left(\begin{array}{ccc}
                       1 & 0&1

\end{array}\right) $. Then we know $d_2^{\,1}(C)=3$ and $C$ is a $2$-symbol MDS code. And $C^\perp$ is the linear code over $\mathbb{F}_2$ with a generator matrix $\left(\begin{array}{ccc}
                       1 & 0&1\\
                       0 & 1&0
\end{array}\right) $. Then $d_2^{\,1}(C^\perp)=2$ and $C^\perp$ is not a $2$-symbol MDS code.

\end{Example}

\section{Linear isomorphisms preserving $b$-symbol weights}
MacWilliams \cite{M} and later Bogart, Goldberg, and Gordon
\cite{BGG} proved that every linear isomorphism preserving Hamming weights between two linear codes over finite fields can be induced by a monomial matrix. Unfortunately, a linear isomorphism induced by a permutation matrix may not preserve $b$-symbol weights of linear codes. In this section,  we  obtain a necessary and sufficient condition for a linear isomorphism preserving $b$-symbol weights between two linear codes over finite fields.

Recall that $n_{r, k}$ is the number of all subspaces of dimension $r$ of a vector space of dimension $k$.
Let ${\rm {\rm PG}}^{r}(\mathbb{F}_q^{k})=\{V_{1}^r,V_{2}^r,\cdots,V_{n_{r,k}}^r \}$ be the set of all subspaces of dimension $r$ of $\mathbb{F}_q^k$. There is a bijection between ${\rm {\rm PG}}^{k-r}(\mathbb{F}_q^{k})$ and ${\rm {\rm PG}}^{r}(\mathbb{F}_q^{k})$, which is defined by
$$
{\rm {\rm PG}}^{k-r}(\mathbb{F}_q^{k})\to {\rm {\rm PG}}^{r}(\mathbb{F}_q^{k}), V^{k-r}\mapsto(V^{k-r})^\bot,\,\, \forall\,\,  V^{k-r}\in {\rm {\rm PG}}^{k-r}(\mathbb{F}_q^{k}).
$$
Hence $n_{r, k}=n_{k-r, k}$. For convenience, if $\frac{k}{2}<r\leq k$, we assume $${\rm {\rm PG}}^{r}(\mathbb{F}_q^{k})=\{V_{1}^{r}=(V_{1}^{k-r})^\bot,V_{2}^{r}=(V_{2}^{k-r})^\bot,\cdots,V_{n_{r, k}}^{r}=(V_{n_{r, k}}^{k-r})^\bot \}.
$$

Let $T_{r,s}$ be an $n_{r,k}\times n_{s,k}$ matrix over the rational number field $\mathbb{Q}$ such that
$$T_{r,s}=(t_{ij})_{n_{r,k}\times n_{s,k}},\,\,\,\,\,t_{ij} = \left\{ \begin{array}{ll}
1,  & \textrm{if $V_i^r\subseteq V_j^s ;$}\\
0,  & \textrm{if $V_i^r\nsubseteq V_j^s ,$}
\end{array} \right.  .$$
 And let $J_{m\times n}$ be the $m\times n$ matrix with all entries being $1$, i.e, $J_{m\times n}=\left(\begin{array}{cccc}
                       1 & \cdots&1\\
                       \vdots&\ddots&\vdots\\
                        1&\cdots&1
\end{array}\right) $.
The following lemma can be found in \cite{LP}.

\begin{lem} \label{T1}
Assume the notations given above, and $1\leq r\leq s\leq z\leq k$. Then
\begin{description}
  \item[(a)] The sum of all rows of $T_{r,s}$ is a constant row vector $\mathbf{t}=n_{r,s}{\bf 1}$.

  \item[(b)] The matrix $T_{1,k-1}$ is an invertible matrix and $T_{1,k-1}^{-1}=\frac{1}{q^{k-2}}(T_{1,k-1}-\frac{q^{k-2}-1}{q^{k-1}-1}J_{n_{1,k}\times n_{1,k}})$ for $k\ge 2$. The sum of all rows of $T_{1,k-1}^{-1}$ is a constant row vector.

  \item[(c)] $T_{r,k-1}T_{1,k-1}=(q^{k-r-1})T_{1,r}^{T}+\frac{q^{k-r-1}-1}{q-1}J_{n_{r,k}\times n_{1,k}}$ and $T_{r,k-1}T_{1,k-1}^{-1}=\frac{1}{q^{r-1}}T_{1,r}^{T}-\frac{q^{r-1}-1}{q^{r-1}(q^{k-1}-1)}J_{n_{r, k}\times n_{1, k}}$ for $k\ge r+1$.

  \item[(d)] $T_{r,s}T_{s,z}=n_{s-r,z-r}T_{r,z}$ for $1\leq r\leq s\leq z\leq k$.
\end{description}
\end{lem}

Let $\varphi$ be an $\mathbb{F}_q$-linear isomorphism from $C$ to $\tilde{C}$, where $C$ and $\tilde{C}$ are two $[n,k]$-linear codes over $\mathbb{F}_q$. Let $G=\left(\begin{array}{c}
                       {\bf g}_{1}\\
                        \cdots \\
                        {\bf g}_{k}

\end{array}\right)$ be a generator matrix of $C$ for some ${\bf g}_{i}\in \mathbb{F}_q^{n}$. Then $\tilde{G}=\left(\begin{array}{c}
                      \varphi( {\bf g}_{1})\\
                        \cdots \\
                        \varphi({\bf g}_{k})

\end{array}\right)$ is a generator matrix of $\tilde{C}$.

\begin{Theorem}\label{6.1}
Assume the notations given above. Then $w_b(\mathbf{c})-w_b(\varphi(\mathbf{c}))$ is constant for any nonzero $\mathbf{c}\in C$ if and only if $\sum_{V\in \Omega_i  }\frac{1}{|V|}(m^b_G(V)-m^b_{\tilde{G}}(V))$ is constant for any $1\leq i \leq n_{1,k}$, where $s=\min\{b,k-1\}$ and $\Omega_i=\{V\in {\rm PG}^{\leq s}(\mathbb{F}_q^k)\,|\,V^1_i\subseteq V\}$.

\end{Theorem}

\begin{proof}
Let $\phi$ be the $\mathbb{F}_q$-linear isomorphism from $\mathbb{F}_q^{k}$ to $ C$ such that $\phi(\mathbf{y})=\mathbf{y}G$ for any $\mathbf{y}\in \mathbb{F}_q^{k}$.
And let $\tilde{\phi}$ be the $\mathbb{F}_q$-linear isomorphism from $\mathbb{F}_q^{k}$ to $ \tilde{C}$ such that $\tilde{\phi}(\mathbf{y})=\mathbf{y}\tilde{G}$ for any $\mathbf{y}\in \mathbb{F}_q^{k}$.
Then $\tilde{\phi}= \varphi\phi$ by the definition of $\tilde{G}$.
For any nonzero $\mathbf{c}\in C$, there is a $\mathbf{y}$ such $\mathbf{c}=\phi(\mathbf{y})$ and $\tilde{\mathbf{c}}=\varphi(\mathbf{c})=\tilde{\phi}(\mathbf{y})$. By Lemma~\ref{r pair weight}, we have
\begin{equation}\label{6.c}
  w_{b}(\mathbf{c})= n-\theta^b_{G}(\langle\mathbf{y}\rangle^{\bot})
\end{equation}
 and \begin{equation}\label{6.d}
  w_{b}(\tilde{\mathbf{c}})= n-\theta^b_{\tilde{G}}(\langle\mathbf{y}\rangle^{\bot}).
\end{equation}

Let $\Delta_{r}=(m^b_{G}(V^r_{1}),m^b_{G}(V^r_{2}),\cdots,m^b_{G}(V^r_{n_{r,k}}))$, $$\tilde{\Delta}_{r}=(m^b_{\tilde{G}}(V^r_{1}),m^b_{\tilde{G}}(V^r_{2}),\cdots,m^b_{\tilde{G}}(V^r_{n_{r,k}}))$$ for $0\leq r \leq \min\{b,k-1\}$ and $\Gamma_{k-1}=(\theta^b_{G}(V^{k-1}_{1}),\theta^b_{G}(V^{k-1}_{2}),\cdots,\theta^b_{G}(V^{k-1}_{n_{1,k-1}})),$ $$\tilde{\Gamma}_{k-1}=(\theta^b_{\tilde{G}}(V^{k-1}_{1}),\theta^b_{\tilde{G}}(V^{k-1}_{2}),\cdots,\theta^b_{\tilde{G}}(V^{k-1}_{n_{1,k-1}})).$$ Assume $s=\min\{b,k-1\}$, then we get
\begin{equation}\label{6.a}
  \Gamma_{k-1}=\sum_{r=0}^{s}\Delta_{r}T_{r,k-1}=m^b_G(\mathbf{0})\textbf{1}+\sum_{r=1}^{s}\Delta_{r}T_{r,k-1}
\end{equation}
and
\begin{equation}\label{6.b}
  \tilde{\Gamma}_{k-1}=\sum_{r=0}^{s}\tilde{\Delta}_{r}T_{r,k-1}=m^b_{\tilde{G}}(\mathbf{0})\textbf{1}+\sum_{r=1}^{s}\tilde{\Delta}_{r}T_{r,k-1}
\end{equation}
by the definition of $\theta^b_{G}$.

Suppose $a=w_b(\mathbf{c})-w_b(\tilde{\mathbf{c}})$ for any nonzero $\mathbf{c}\in C$. By Equation~\ref{6.c} and Equation~\ref{6.d}, we have $\theta^b_{G}(\langle\mathbf{y}\rangle^{\bot})-\theta^b_{\tilde{G}}(\langle\mathbf{y}\rangle^{\bot})=-a$ for any nonzero $\mathbf{y}\in \mathbb{F}_q^{k}$ and $$\Gamma_{k-1}-\tilde{\Gamma}_{k-1}=-a\mathbf{1}.$$

By Equation~\ref{6.a} and Equation~\ref{6.b}, we have $$\sum_{r=1}^{s}(\Delta_{r}-\tilde{\Delta}_{r})T_{r,k-1}=(m^b_{\tilde{G}}(\mathbf{0})-m^b_G(\mathbf{0})-a)\textbf{1}$$ and $$\sum_{r=1}^{s}(\Delta_{r}-\tilde{\Delta}_{r})T_{r,k-1}T_{1,k-1}^{-1}=\frac{m^b_{\tilde{G}}(\mathbf{0})-m^b_G(\mathbf{0})-a}{n_{1,k-1}}\mathbf{1}.$$
Then we have $$q\sum_{V\in \Omega_i  }\frac{1}{|V|}(m^b_G(V)-m^b_{\tilde{G}}(V))-\sum_{r=2}^s\sum_{V^r\in {\rm PG}^r(\mathbb{F}_q^k)}\frac{q^{r-1}-1}{q^{r-1}(q^{k-1}-1)}(m^b_G(V^r)-m^b_{\tilde{G}}(V^r))=\frac{m^b_{\tilde{G}}(\mathbf{0})-m^b_G(\mathbf{0})-a}{n_{1,k-1}},$$
since the element in the $i$th position of the vector $\sum_{r=1}^{s}(\Delta_{r}-\tilde{\Delta}_{r})T_{r,k-1}T_{1,k-1}^{-1}$ is
$$q\sum_{V\in \Omega_i  }\frac{1}{|V|}(m^b_G(V)-m^b_{\tilde{G}}(V))-\sum_{r=2}^s\sum_{V^r\in {\rm PG}^r(\mathbb{F}_q^k)}\frac{q^{r-1}-1}{q^{r-1}(q^{k-1}-1)}(m^b_G(V^r)-m^b_{\tilde{G}}(V^r))$$

by Lemma~\ref{T1} (c), where $\Omega_i=\{V\in {\rm PG}^{\leq s}(\mathbb{F}_q^k)\,|\,V^1_i\subseteq V\}$.
Hence $$q\sum_{V\in \Omega_i  }\frac{1}{|V|}(m^b_G(V)-m^b_{\tilde{G}}(V))=\sum_{r=2}^s\sum_{V^r\in {\rm PG}^r(\mathbb{F}_q^k)}\frac{q^{r-1}-1}{q^{r-1}(q^{k-1}-1)}(m^b_G(V^r)-m^b_{\tilde{G}}(V^r))+\frac{m^b_{\tilde{G}}(\mathbf{0})-m^b_G(\mathbf{0})-a}{n_{1,k-1}}$$
and $\sum_{V\in \Omega_i  }\frac{1}{|V|}(m^b_G(V)-m^b_{\tilde{G}}(V))$ is constant for any $1\leq i \leq n_{1,k}$.

Suppose $\sum_{V\in \Omega_i  }\frac{1}{|V|}(m^b_G(V)-m^b_{\tilde{G}}(V))=b$ for any $1\leq i \leq n_{1,k}$.
Then $$\sum_{r=1}^{s}(\Delta_{r}-\tilde{\Delta}_{r})T_{r,k-1}T_{1,k-1}^{-1}$$ and $\sum_{r=1}^{s}(\Delta_{r}-\tilde{\Delta}_{r})T_{r,k-1}$ are constant vectors by Lemma~\ref{T1} (b),
since the element in the $i$th position of the vector $\sum_{r=1}^{s}(\Delta_{r}-\tilde{\Delta}_{r})T_{r,k-1}T_{1,k-1}^{-1}$ is
$$q\sum_{V\in \Omega_i  }\frac{1}{|V|}(m^b_G(V)-m^b_{\tilde{G}}(V))-\sum_{r=2}^s\sum_{V^r\in {\rm PG}^r(\mathbb{F}_q^k)}\frac{q^{r-1}-1}{q^{r-1}(q^{k-1}-1)}(m^b_G(V^r)-m^b_{\tilde{G}}(V^r))$$
$$=qb-\sum_{r=2}^s\sum_{V^r\in {\rm PG}^r(\mathbb{F}_q^k)}\frac{q^{r-1}-1}{q^{r-1}(q^{k-1}-1)}(m^b_G(V^r)-m^b_{\tilde{G}}(V^r)).$$
By Equation~\ref{6.a} and Equation~\ref{6.b}, we know $$\Gamma_{k-1}-\tilde{\Gamma}_{k-1}=\sum_{r=1}^{s}(\Delta_{r}-\tilde{\Delta}_{r})T_{r,k-1}+(m^b_{\tilde{G}}(\mathbf{0})-m^b_G(\mathbf{0}))\textbf{1}$$ is a constant vector and $\theta^b_{G}(\langle\mathbf{y}\rangle^{\bot})-\theta^b_{\tilde{G}}(\langle\mathbf{y}\rangle^{\bot})$ is constant for any nonzero $\mathbf{y}\in \mathbb{F}_q^{k}$. Hence $w_b(\mathbf{c})-w_b(\varphi(\mathbf{c}))$ is constant for any nonzero $\mathbf{c}\in C$ by Equation~\ref{6.c} and Equation~\ref{6.d}.
\end{proof}

\begin{Corollary}\label{6.2}
Assume the notations given above. Then $w_b(\mathbf{c})=w_b(\varphi(\mathbf{c}))$ for any $\mathbf{c}\in C$ if and only if there exists $\mathbf{c_0}\in C$ such that $w_b(\mathbf{c_0})=w_b(\varphi(\mathbf{c_0}))$ and $$\sum_{V\in \Omega_i  }\frac{1}{|V|}(m^b_G(V)-m^b_{\tilde{G}}(V))$$ is constant for any $1\leq i \leq n_{1,k}$, where $s=\min\{b,k-1\}$ and $\Omega_i=\{V\in {\rm PG}^{\leq s}(\mathbb{F}_q^k)\,|\,V^1_i\subseteq V\}$.
\end{Corollary}

There is an example for using Corollary~\ref{6.2} to determine a linear isomorphism preserving $b$-symbol weights when $b=2$ in Section 5 of \cite{LP}.
When $b=1$, we obtain the classical MacWilliams extension theorem \cite{M}\cite{BGG} in the next corollary.
\begin{Corollary}\label{6.3}
Assume the notations given above. Then $w_1(\mathbf{c})=w_1(\varphi(\mathbf{c}))$ for any $\mathbf{c}\in C$ if and only if there exists an monomial matrix $M$ such that $\varphi(\mathbf{x})=\mathbf{x}M$ for any $\mathbf{x}\in \mathbb{F}_q^n$.

\end{Corollary}

\begin{proof}
When $b=1$,  we have $w_1(\mathbf{c})=w_1(\varphi(\mathbf{c}))$ for any $\mathbf{c}\in C$ if and only if $m^1_G(V_i^1)=m^1_{\tilde{G}}(V_i^1)$ for any $1\leq i\leq n_{1,k}$ by Corollary~\ref{6.2} and $$\sum_{V\in \Omega_i  }\frac{1}{|V|}(m^1_G(V)-m^1_{\tilde{G}}(V))=m^1_G(V_i^1)-m^1_{\tilde{G}}(V_i^1)$$ for any $1\leq i\leq n_{1,k}$.
Hence $w_1(\mathbf{c})=w_1(\varphi(\mathbf{c}))$ for any $\mathbf{c}\in C$ if and only if there exists an monomial matrix $M$ such that $\varphi(\mathbf{x})=\mathbf{x}M$ for any $\mathbf{x}\in \mathbb{F}_q^n$ by using the definitions of the functions $m^1_{G}$ and $m^1_{\tilde{G}}$.
\end{proof}

From Theorem \ref{6.1}, we know that if we want to determine a linear isomorphism is preserving $b$-symbol weights of linear codes or not, it is crucial to calculate the value  $\sum_{V\in \Omega_i  }\frac{1}{|V|}m^b_G(V)$ for an $[n,k]$-linear code $C$ with a generator matrix $G$, where $s=\min\{b,k-1\}$ and $\Omega_i=\{V\in {\rm {\rm PG}}^{\leq s}(\mathbb{F}_q^k)\,|\,V^1_i\subseteq V\}$.

Recall that we have assumed that $G=(G_{0},\cdots,G_{n-1})$ is a generator matrix of an $[n,k]$-linear code $C$ over $\mathbb{F}_q$.
Then we assume $$S_j=\mathbb{F}_qG_{j}+\mathbb{F}_qG_{j+1}+\cdots+\mathbb{F}_qG_{j+b-1}$$ which is a $\mathbb{F}_q$-subspace of $\mathbb{F}_q^k$ and $$\tilde{S}_j=( G_{j},G_{j+1},\cdots,G_{j+b-1})$$ is a $k\times b$ submatrix of $G$ for $0\leq j\leq n-1$.
Also we know that $\dim(S_j)=rank(\tilde{S}_j)$.

\begin{Theorem}\label{611}
Assume $\kappa_{ij}=\left\{ \begin{array}{ll}
1,  & \textrm{if $V^1_{i}\subseteq S_j ;$}\\
0,  & \textrm{if $V^1_{i}\nsubseteq S_j .$}
\end{array} \right.$ for $1\leq i\leq n_{1,k}$ and $1\leq j \leq n$, $\sum_{V\in \Omega_i  }\frac{1}{|V|}m^b_G(V)=\sum_{j=1}^{n}\kappa_{ij}q^{-rank(\tilde{S}_j)}$.
\end{Theorem}

\begin{proof}
It is easy to prove this lemma by using the definition of the function $m^b_{G}$.
\end{proof}

\begin{Remark}\label{5.6}
Let $C$ be an $[n,k]$-linear code over $\mathbb{F}_q$ with a generator matrix $G=(G_{0},\cdots,G_{n-1})$, then we calculate $f_i=\sum_{j=1}^{n}\kappa_{ij}q^{-rank(\tilde{S}_j)}$ for $1\leq i\leq n_{1,k}$ by the following steps.
First we can calculate $\{S_0,S_1,\cdots,S_{n-1}\}$ and $|{\rm {\rm PG}}^1(S_i)|\leq \frac{q^b-1}{q-1}$.
Assume $T=\bigcup_{i=1}^{n}{\rm {\rm PG}}^1(S_i)$, we have $|T|\leq n\frac{q^b-1}{q-1}$.
If $V^1_i\notin T$, then $f_i=0$ by Theorem~\ref{611}. So we only need to calculate $f_i$ for $|T|$ subspaces of one dimension of $\mathbb{F}_q^k$.

However, if we simply check $b$-symbol weights of all the codewords of $C$ and $\tilde{C}$, then we need to calculate $2\cdot\frac{q^k-1}{q-1}$ subspaces of dimension one of $\mathbb{F}_q^k$ for their $b$-symbol weights since $\mathbf{c}$ and $\lambda\mathbf{c}$ have same $b$-symbol weight for $\mathbf{c}\in C$ and $\lambda\in \mathbb{F}_q^*$.
So the determination of a linear isomorphism preserve $b$-symbol weight by using Theorem \ref{6.1} is more efficient, since $2\cdot|T|\leq 2n\cdot \frac{q^b-1}{q-1}<<2\cdot\frac{q^k-1}{q-1}$ when $k>>b$. For example, when $C$ is a $[10,6]$-linear code $C$ over $\mathbb{F}_{31}$ and $b=3$, then $2\cdot|T|\leq 19860$ is much less than $2\cdot\frac{31^6-1}{31-1}=59166912$.

\end{Remark}

\vskip 4mm

\noindent {\bf Acknowledgement.}
This work was supported by NSFC (Grant No. 11871025).



\begin{thebibliography}{99}




\bibitem{B}	 Beelen P.: A note on the generalized Hamming weights of Reed-Muller codes. Applicable Algebra in Engineering, Communication and Computing 30(3), 233-242 (2019).


\bibitem{BGG}  Bogart K., Goldberg D., Gordon J.: An elementary proof of the MacWilliams theorem on equivalence of codes. Information and Control 37(1), 19-22 (1978).

\bibitem{CB} Cassuto Y., Blaum M.: Codes for symbol-pair read channels. IEEE Transactions on Information Theory 57(12), 8011-8020 (2011).

\bibitem{C} Chee Y. M., Ji L., Kiah H. M., Wang C.,Yin J.: Maximum distance separable codes for symbol-pair read channels. IEEE Transactions on Information Theory 59(11), 7259-7267 (2013).

\bibitem{CLL} Chen B., Lin L., Liu H.: Constacyclic symbol-pair codes: lower bounds and optimal constructions. IEEE Transactions on Information Theory 63(12), 7661-7666 (2017).

\bibitem{CJK} Chee Y. M., Ji L., Kiah H. M., Wang C., Yin J.: Maximum distance separable codes for symbol-pair read channels. IEEE Transactions on Information Theory 59(11), 7259-7267 (2013).

\bibitem{D} Ding	B., Ge G., Zhang J., Zhang T., Zhang Y.: New constructions of MDS symbol-pair codes. Designs, Codes Cryptography 86(4), 841-859 (2018).


\bibitem{DZ} Ding B., Zhang T., Ge G.: Maximum distance separable codes for b-symbol read channels. Finite Fields Their Applications 49, 180-197 (2018).

\bibitem{DNSS}Dinh H. Q., Nguyen B. T., Singh A. K., Sriboonchitta S.: On the symbol-pair distance of repeated-root constacyclic codes of prime power lengths. IEEE Transactions on Information Theory 64(4), 2417-2430 (2017).

\bibitem{DWLS}	 Dinh H. Q., Wang X., Liu H., Sriboonchitta S.: On the symbol-pair distances of repeated-root constacyclic codes of length $2p^s$. Discrete Mathematics 342(11), 3062-3078 (2019).

\bibitem{FL} Fan Y., Liu H.: Generalized Hamming equiweight linear codes. Acta Electronica Sinica 31(10), 1591-1593 (2003).

\bibitem{H} Huffman C. W., Pless V.: Fundamentals of error-correcting codes. Cambridge University Press (2003).

\bibitem{KZL} Kai X., Zhu S., Li P.: A construction of new MDS symbol-pair codes. IEEE Transactions on Information Theory 61(11), 5828-5834 (2015).


\bibitem{JFW} Jian G., Feng R., Wu H.: Generalized Hamming weights of three classes of linear codes. Finite Fields and Their Applications 45(5), 341-354 (2017).

\bibitem{KZL} Kai X., Zhu S., Li P.: A construction of new MDS symbol-pair codes. IEEE Transactions on Information Theory 61(11), 5828-5834 (2015).

\bibitem{KZZLC} Kai X., Zhu S., Zhao Y., Luo H., Chen Z.: New MDS symbol-pair codes from repeated root codes. IEEE Communications Letters 22(3) 462-465 (2018).

\bibitem{LG} Li S., Ge G.: Constructions of maximum distance separable symbol-pair codes using cyclic and constacyclic codes. Designs, Codes and Cryptography 84(3), 359-372 (2017).

\bibitem{LP}  Liu H., Pan X.: Generalized pair weights of linear codes and linear isomorphisms preserving pair weights. IEEE Transactions on Information Theory 68(1), 105-117 (2022).

\bibitem{M} MacWilliams J.: A theorem on the distribution of weights in a systematic code. Bell System Technical Journal 42(1), 79-94 (1963).

\bibitem{ML1} Ma J., Luo J.: MDS symbol-pair codes from repeated-root cyclic codes. Designs, Codes and Cryptography 90, 121-137 (2022)

\bibitem{ML2}  Ma J., Luo J.: New MDS symbol-pair codes from repeated-root cyclic codes over Finite Field. arXiv.org (2020).


\bibitem{OW} Ozarow L. H., Wyner A. D.: Wire-tap channel II," AT$\&$T Bell Laboratories Technical Journal. 63(10), 2135-2157 (1984).


\bibitem{ST} Storme L., Thas J. A.: M.D.S. codes and arcs in ${\rm PG}(n,q)$ with $q$ even: an improvement of the bounds of Bruen, Thas, and Blokhuis. Journal of Combinatorial Theory 62(1), 139-154 (1993).

\bibitem{TV} Tsfasman M. A., Vladut S. G.: Geometric approach to higher weights. IEEE Transactions on Information Theory 41(6), 1564-1588 (1995).

\bibitem{W} Wei V. K.: Generalized Hamming weights for linear codes. IEEE Transactions on information theory 37(5), 1412-1418 (1991).




\bibitem{YBS}Yaakobi E., Bruck J., Siegel P. H.: Constructions and decoding of cyclic codes over $ b $-symbol read channels. IEEE Transactions on Information Theory 62(4), 1541-1551 (2016).

\bibitem{YL} Yang M., Li J., Feng K.: Construction of cyclic and constacyclic codes for $b$-symbol read channels meeting the Plotkin-like bound. arXiv.org (2016).


\bibitem{YF} Yang M., Li J., Feng K., Lin D.: Generalized Hamming weights of irreducible cyclic codes. IEEE Transactions on Information Theory 61(9), 4905-4913 (2015).












\end{thebibliography}
\end{document}